\documentclass[12pt]{article}
\usepackage[margin=1in]{geometry}
\usepackage{amssymb,amsthm,amsmath}
\usepackage{amsfonts}       
\usepackage{tikz-cd}
\usetikzlibrary{arrows,decorations.pathmorphing,backgrounds,positioning,fit,petri,calc,shapes.misc,decorations.markings,intersections}
\usepackage[affil-it]{authblk}
\usepackage{enumerate}
\usepackage[numbers]{natbib}
\usepackage[all]{xy}
\usepackage{setspace}
\usepackage{float}
\usepackage{titling}
\newcommand{\subtitle}[1]{%
  \posttitle{%
    \par\end{center}
    \begin{center}\large#1\end{center}
    \vskip0.5em}%
}
\makeatletter
\newcommand{\oset}[3][0ex]{%
  \mathrel{\mathop{#3}\limits^{
    \vbox to#1{\kern-2\ex@
    \hbox{$\scriptstyle#2$}\vss}}}}
\makeatother

\newtheorem{thm}{Theorem}

\newtheorem{prop}{Proposition}

\newtheorem{cor}{Corollary}

\theoremstyle{definition}
\newtheorem{definition}{Definition}
\theoremstyle{definition}

\usepackage{tikz-cd}
\usepackage{tikz}
\usepackage{tikz-3dplot}

\providecommand{\norm}[1]{\lVert#1\rVert}

\providecommand{\R}{\mathbb{R}}

\title{Quantization as a Categorical Equivalence}
\author{Benjamin H.~Feintzeig\\ \texttt{bfeintze@uw.edu}}
\affil{{Department of Philosophy}\\{University of Washington}}
\date{}

\begin{document}

\maketitle

\begin{abstract}
    We demonstrate that, in certain cases, quantization and the classical limit provide functors that are ``almost inverse" to each other.  These functors map between categories of algebraic structures for classical and quantum physics, establishing a categorical equivalence.
\end{abstract}

\section{Introduction}

The purpose of this paper is to develop tools for assessing the extent to which mathematical models of classical and quantum physics share common structure.  To do this, we analyze functors between categories of classical and quantum models.  It is well known \citep{BaBaDo04} that different properties of functors provide information about relations between models in different categories.  Previously, Landsman \citep{La02a} has proposed understanding quantization as a functor;\footnote{Others \citep{Ri93,HoRiSc08,BiGa15} have also provided results concerning senses in which quantization is a functor.} however, he does not analyze the properties of this functor.  The current paper analyzes the properties of quantization functors that bear on the shared structure of classical and quantum physics.

One standard for when a functor demonstrates structural equivalence between mathematical models is when the functor exhibits a \emph{categorical equivalence} between the corresponding categories of models.  This is the standard we appeal to in this paper.  We show that a quantization functor can be supplemented by a classical limit functor, which serves as an ``almost inverse," and which together yield a categorical equivalence between certain categories of models of classical and quantum physics.

Unlike Landsman \citep{La01c,La01a,La02a}, who considers categories whose arrows are Hilbert bimodules, instead we define categories whose arrows are certain types of *-homomorphisms that directly preserve the algebraic structure of a model.  We take this as a first step for defining the actions of quantization and classical limit functors on arrows, although it forces us to significantly restrict the arrows in the categories we consider.  It would be much more desirable to extend the functoriality of quantization, and the categorical equivalence, to more interesting morphisms than the ones considered here.  The current paper thus should be understood as a proof of concept, which we hope to apply in future work to the categories Landsman considers by defining the classical limit of a Hilbert bimodule.  The main result of this paper is thus the framework of using the classical limit functor to establish that quantization is a categorical equivalence.

In what follows, we investigate the quantization of two types of models of classical physics.  In \S\ref{sec:quant}, we introduce background on strict deformation quantization of C*-algebras and the classical limit.  In \S\ref{sec:wquant} we analyze the quantization of the C*-Weyl algebra, which applies to physical systems with linear phase spaces.  In \S\ref{sec:rquant} we analyze Rieffel's quantization, which applies to physical systems, whose phase space carries an action of $\mathbb{R}^d$.  In each case, we establish that corresponding quantization and classical limit functors exhibit a categorical equivalence.  In \S\ref{sec:con}, we conclude with some discussion of the results and future directions.

\section{Quantization and the Classical Limit}
\label{sec:quant}

We will investigate two functors $Q$ and $L$ corresponding to the processes of quantization and the classical limit.  We will present general definitions for each process and then establish that for certain classes of physical systems---corresponding to particular categories of models of classical and quantum physics---these functors form a categorical equivalence.

To quantize a model of classical physics, we begin with a commutative C*-algebra of functions on a phase space and continuously deform the product operation to arrive at a noncommutative C*-algebra of bounded operators on a Hilbert space.  The resulting family of C*-algebras indexed by the parameter $\hbar$ forms a structure called a continuous bundle.\footnote{For more details on continuous bundles of C*-algebras, see Dixmier \citep[][Ch. 10]{Di77}, Kirchberg and Wasserman \citep{KiWa95}, or Landsman \citep[][\S II.1.2]{La98b}.}

\begin{definition}[continuous bundle of C*-algebras]

A \emph{(uniformly)\footnote{See Steeger and Feintzeig \citep[][Appendix B]{StFe21a} for more details on different continuity conditions for bundles of C*-algebras.  All bundles considered in this paper are uniformly continuous.} continuous bundle of C*-algebras} over a locally compact \emph{base space}\footnote{It is possible to include more general locally compact metric spaces as base spaces.  See Steeger and Feintzeig \citep{StFe21a}.} $I\subseteq \mathbb{R}$ (where $I$ contains $0$ as an accumulation point) is: 
\begin{itemize}
    \item a family of C*-algebras $\{\mathfrak{A}_\hbar\}_{\hbar\in I}$ called \emph{fibers};
    \item a C*-algebra $\mathfrak{A}$ of \emph{continuous sections}; and 
    \item a family of surjective *-homomorphisms $\{\phi_\hbar: \mathfrak{A}\to\mathfrak{A}_\hbar\}_{\hbar\in I}$ called \emph{evaluation maps}.
\end{itemize}
Together, these structures must satisfy for each $a\in\mathfrak{A}$,
\begin{enumerate}[(i)]
\item $\norm{a} = \sup_{\hbar\in I}\norm{\phi_\hbar(a)}_\hbar$, where $\norm{\cdot}_\hbar$ denotes the norm on the fiber algebra $\mathfrak{A}_\hbar$;
\item for each uniformly continuous and bounded function $f: I\to \mathbb{C}$, there is a section $fa\in \mathfrak{A}$ such that $\phi_\hbar(fa) = f(\hbar)\phi_\hbar(a)$;
\item the map $\hbar\mapsto \norm{\phi_\hbar(a)}_\hbar$ is uniformly continuous and bounded.
\end{enumerate}
\end{definition}

Such continuous bundles can be constructed by deforming the product in the direction of the Poisson bracket of a classical phase space.  One can do so with the following notion of a quantization map.

\begin{definition}[strict deformation quantization]
A \emph{strict deformation quantization} (over $I\subseteq \mathbb{R}$ with $0\in I$) of a manifold $M$ with a *-algebra $\mathcal{P}\subseteq C_b(M)$ of continuous, bounded functions carrying a Poisson bracket is:
\begin{itemize}
    \item a family of C*-algebras $\{\mathfrak{A}_\hbar\}_{\hbar\in I}$; and
    \item a family of linear \emph{quantization maps} $\{\mathcal{Q}_\hbar: \mathcal{P}\to\mathfrak{A}_\hbar\}_{\hbar\in I}$, where $\mathcal{Q}_0$ is the identity.
\end{itemize}
Together, these structures must satisfy for each $f,g\in\mathcal{P}$,
\begin{enumerate}[(i)]
    \item (von Neumann's condition) $\lim_{\hbar\to 0}\norm{\mathcal{Q}_\hbar(f)\mathcal{Q}_\hbar(g) - \mathcal{Q}_\hbar(fg)}_\hbar = 0$;
    \item (Dirac's condition) $\lim_{\hbar\to 0}\norm{\frac{i}{\hbar}\big[\mathcal{Q}_\hbar(f),\mathcal{Q}_\hbar(g)\big] - \mathcal{Q}_\hbar\big(\{f,g\}\big)}_\hbar = 0$;
    \item (Rieffel's condition) the map $\hbar\mapsto \norm{\mathcal{Q}_\hbar(f)}_\hbar$ is continuous on $I$;
    \item (Deformation condition) for each $\hbar\in I$, the map $\mathcal{Q}_\hbar$ is injective, its image $\mathcal{Q}_\hbar[\mathcal{P}]$ is closed under the product in $\mathfrak{A}_\hbar$, and $\mathcal{Q}_\hbar[\mathcal{P}]$ is dense in $\mathfrak{A}_\hbar$.
\end{enumerate}
\end{definition}

\noindent Every strict deformation quantization satisfying mild technical conditions defines a continuous bundle of C*-algebras (See Landsman \citep[][\S II.1.2]{La98b} or Steeger and Feintzeig \citep[][Appendix A]{StFe21a}).  The algebra of sections is generated by the maps $[\hbar\mapsto \mathcal{Q}_\hbar(f)]$ for each $f\in \mathcal{P}$ and the maps $\phi_\hbar$ are given concretely as evaluation of the sections at a particular value $\hbar\in  I$.

In the opposite direction of quantization, the classical limit can be understood as the process of restricting a continuous bundle of C*-algebras obtained from quantization back to the commutative C*-algebra at $\hbar = 0$.  Steeger and Feintzeig \citep{StFe21a} show that the fiber algebra $\mathfrak{A}_0$ at $\hbar = 0$ can be reconstructed from a given continuous bundle of C*-algebras $((\mathfrak{A}_\hbar,\phi_\hbar)_{\hbar\in  I},\mathfrak{A})$ over $I = (0,1]$ containing only information about the quantum theory for $\hbar>0$. To do so, consider the closed two-sided ideal $K_0 = \{a\in\mathfrak{A}\ |\ \lim_{\hbar\to 0}\norm{\phi_\hbar(a)}_\hbar = 0\}$ of sections vanishing at $\hbar = 0$. Steeger and Feintzeig \citep[][\S4]{StFe21a} show that the quotient
\begin{align}
\label{eq:alglim}
    \mathfrak{A}_0 = \mathfrak{A}/K_0
\end{align}
is the unique limit point C*-algebra at $\hbar\to 0$ of the bundle up to *-isomorphism and that the quotient map $\phi_0: \mathfrak{A}\to \mathfrak{A}/K_0$ defines the unique evaluation map at the fiber over $\hbar = 0$.  This procedure allows one to reconstruct the fiber algebra $\mathfrak{A}_0$ of the classical theory at $\hbar = 0$ from the bundle of quantum algebras for $\hbar>0$. 

We now have enough tools to define the action of quantization and the classical limit on objects.  Quantization associates a Poisson algebra $\mathcal{P}$ of functions on $M$ to a non-commutative C*-algebra $\mathfrak{A}_\hbar$ obtained by strict deformation quantization.  On the other hand, the classical limit associates a non-commutative C*-algebra $\mathfrak{A}_\hbar$ to the unique commutative algebra $\mathfrak{A}_0$ obtained as the $\hbar\to 0$ limit of some continuous bundle, which is an algebra of functions on the phase space.  In order to draw structural comparisons, we further need a way to associate morphisms of classical and quantum models with one another.

One can take the classical limit of a morphism as follows.  Consider two continuous bundles of C*-algebras $((\mathfrak{A}_\hbar,\phi_\hbar)_{\hbar\in I},\mathfrak{A})$ and $((\mathfrak{B}_\hbar,\psi_\hbar)_{\hbar\in I},\mathfrak{B})$ over $I = (0,1]$ representing quantum systems for $\hbar>0$.  Suppose one has a family of morphisms $(\alpha_\hbar: \mathfrak{A}_\hbar\to\mathfrak{B}_\hbar)_{\hbar\in I}$ of the fibers for $\hbar >0$ that lift to a *-homomorphism $\alpha:\mathfrak{A}\to \mathfrak{B}$ of the algebras of sections commuting with the evaluation maps in the sense that
\begin{align}
\label{eq:morphlift}
    \alpha_\hbar\circ\phi_\hbar = \psi_\hbar\circ \alpha 
\end{align}
for each $\hbar>0$.  Steeger and Feintzeig \citep[][\S5]{StFe21a} show that in this situation the morphism is appropriately continuous in $\hbar$ so that there is a unique limit morphism $\alpha_0: \mathfrak{A}_0\to \mathfrak{B}_0$ obtained by factoring through the quotient $\mathfrak{A}/K_0$ and thus satisfying
\begin{align}
\label{eq:morphlim}
\alpha_0\circ\phi_0 = \psi_0\circ \alpha.
\end{align}
This provides a direct way to associate morphisms of a model of quantum physics with morphisms of a model of classical physics through the classical limit.

Now we will encode the conditions under which one can take the classical limit of a morphism of a fiber at a given value $\hbar>0$.  We will associate with each family of quantization maps $\{\mathcal{Q}_\hbar\}_{\hbar\in I}$ a collection of \emph{rescaling maps} $\big\{R^\mathcal{Q}_{\hbar\to\hbar'}: \mathcal{Q}_\hbar[\mathcal{P}]\to\mathcal{Q}_{\hbar'}[\mathcal{P}]\big\}_{\hbar,\hbar'\in I}$ defined by
\begin{align}
    R^\mathcal{Q}_{\hbar\to\hbar'} = \mathcal{Q}_{\hbar'}\circ(\mathcal{Q}_\hbar)^{-1}
\end{align}
for any $\hbar,\hbar'\in I$.

\begin{definition}[morphisms]
\label{def:morph}
Suppose one has two strict deformation quantizations $(\mathfrak{A}_\hbar,\mathcal{Q}_\hbar)_{\hbar\in  I}$ and $(\mathfrak{B}_\hbar,\mathcal{Q}'_\hbar)_{\hbar\in  I}$ over $I \subseteq \mathbb{R}$ of $\mathcal{P}$ and $\mathcal{P}'$, respectively.  A *-homomorphism $\alpha_\hbar: \mathfrak{A}_\hbar\to\mathfrak{B}_\hbar$ between the fiber algebras at a fixed value $\hbar\neq 0\in  I$ is called
\begin{enumerate}[(i)]
\item \emph{smooth} if $\alpha_\hbar\big[\mathcal{Q}_\hbar[\mathcal{P}]\big]\subseteq \mathcal{Q}'_\hbar[\mathcal{P}']$;
\item \emph{scaling} if for every $\hbar' \neq 0$, the map
\begin{align}
    \alpha_{\hbar'} = R^{\mathcal{Q}'}_{\hbar\to\hbar'} \circ \alpha_\hbar\circ R^\mathcal{Q}_{\hbar'\to\hbar}
\end{align}
extends continuously to a *-homomorphism $\mathfrak{A}_{\hbar'}\to\mathfrak{B}_{\hbar'}$.
\end{enumerate}
\end{definition}

\noindent The smoothness condition says that a morphism preserves the additional structure of the collection of quantized smooth functions on which the Poisson bracket is defined.  Insofar as the information that certain quantities are smooth, in addition to being merely continuous, is part of the physical theory, structure-preserving maps between quantum models should encode this structure.  The scaling condition says that a morphism preserves the algebraic structure regardless of the numerical value of $\hbar$, where the rescaling maps are used to shift the morphism to different values of $\hbar'>0$ in order to make the comparison.  Insofar as the numerical value of Planck's constant $\hbar$ in a strict deformation quantization depends on a system of units (See Feintzeig \citep{Fe20}), it is merely a conventional choice, which the status of a map as preserving the structure of the model should not depend on.  Indeed, since it is typical in a strict deformation quantization---and it holds true in all cases considered in this paper---that all of the algebras $\mathfrak{A}_\hbar$ for $\hbar>0$ are *-isomorphic, we can understand all of the algebras for the quantum theory in different systems of units (different values of $\hbar>0$) as having the same structure.  By understanding morphisms of a quantum theory as *-homomorphisms satisfying the scaling condition, we are only requiring that these structure-preserving maps respect this structural sameness in different systems of units.\footnote{In fact, we leave it as an open question whether there even exist morphisms between fiber C*-algebras of a strict deformation quantization that do not satisfy the scaling condition in cases of interest.  We have not been able to find morphisms between the fiber C*-algebras used in \S\ref{sec:wquant}-\ref{sec:rquant} that fail the scaling condition.}

If a *-homomorphism $\alpha_\hbar: \mathfrak{A}_\hbar\to\mathfrak{B}_\hbar$ is scaling, then the construction surrounding Eq. (\ref{eq:morphlift}) provides a lift of the family $(\alpha_\hbar)_{\hbar\in I}$ to a morphism of the algebra of sections of the bundle and produces a unique limit morphism $\alpha_0$ satisfying Eq. (\ref{eq:morphlim}).  If $\alpha_\hbar$ is smooth, then it follows that $\alpha_0$ preserves the privileged Poisson subalgebra \emph{and} the Poisson bracket defined on it \citep[][Prop. 5.5]{StFe21a}.

We pause here to consider a possible objection to the conditions we require on morphisms in Def. \ref{def:morph}.  It is well-known that some quantizations on compact phase spaces lead to situations where the quantization maps are only defined for discrete values of $\hbar$.  So one might worry that the scaling condition is too restrictive and may not be satisfied for all values of $\hbar>0$.  In the rest of this paper, we will restrict attention to quantizations where $\hbar$ takes continuous values in $(0,1]$ for convenience. So in our examples, the scaling condition does require that morphisms are translatable to all continuous values of $\hbar >0$.  However, our definition is general enough to accommodate other base spaces, even ones with only discrete values of $\hbar$.  In general, one may define strict quantizations and continuous bundles of C*-algebras for locally compact base spaces.  Moreover, the construction of the classical limit from Steeger \& Feintzeig \citep{StFe21a} also applies for more general locally compact base spaces, as long as the base space carries the additional structure of a metric.  In situations where the base space is different than $(0,1]$, the scaling condition only imposes a requirement for those values of $\hbar\neq 0$ in the base space on which the quantization maps are defined.  The central idea of the scaling condition is only to require that morphisms be translatable among all of the allowed values of $\hbar$ by the rescaling maps.  As long as morphisms can be translated to all the values of $\hbar$ in the base space (whether discrete or continuous), those morphisms can be lifted and extended to a limit point of the base space by the construction of Steeger and Feintzeig.  As such, the scaling condition is a plausible condition on morphisms, and indeed we do not know of any morphisms that fail to satisfy the scaling condition.

Now that we have tools for understanding quantization and the classical limit, we will proceed to characterize two categories of models of classical physics that can be quantized functorially, and whose quantization we will demonstrate provides a categorical equivalence.

\section{The C*-Weyl Algebra for Linear Phase Spaces}
\label{sec:wquant}

One standard method for quantizing classical theories via the C*-Weyl algebra applies to systems whose phase space is the dual space (i.e., the collection of continuous linear functionals) $V'$ of a topological vector space $V$ with a symplectic form $\sigma$ (i.e., a non-degenerate, bilinear, antisymmetric map $V\times V\to\mathbb{R}$).  In this case, the Poisson *-algebra $\Delta(V,0) \subseteq C_b(V')$ is generated by linear combinations of the functions $W_0(f): V'\to\mathbb{C}$ for each fixed $f\in V$ defined by
\begin{align}
W_0(f)(F) = e^{iF(f)}
\end{align}
for all $F\in V'$.  The Poisson bracket on $\Delta(V,0)$ is defined by the linear extension of
\begin{align}
    \{W_0(f),W_0(g)\} = \sigma(f,g)W_0(f+g)
\end{align}
for all $f,g\in V$.  This algebra $\Delta(V,0)$ is norm dense in the C*-algebra $AP(V')$ of continuous almost periodic functions on the phase space $V'$.  This structure specifies the classical model.

The corresponding quantum model is obtained through the exponentiated Weyl form of the canonical commutation relations, which define for each $\hbar>0$ a C*-algebra $\mathcal{W}(V,\hbar\sigma)$.  A dense *-subalgebra $\Delta(V,\hbar\sigma)$ is generated freely by linearly independent elements of the form $W_\hbar(f)$ for $f\in V$ with multiplication and involution operations specified by
\begin{align}
    W_\hbar(f)W_\hbar(g) &= e^{-\frac{i}{2}\hbar\sigma(f,g)}W_\hbar(f+g)\\
    W_\hbar(f)^* &= W_\hbar(-f)
\end{align}
for all $f,g\in V$.  There is a unique maximal C*-norm on $\Delta(V,\hbar\sigma)$ and the C*-Weyl algebra $\mathcal{W}(V,\hbar\sigma) = \overline{\Delta(V,\hbar\sigma)}$ is defined as the completion of this dense *-subalgebra with respect to the C*-norm \citep[See][]{Sl72,MaSiTeVe74,Pe90,BiHoRi04a}.

In the special case where $V = \mathbb{R}^{2n}$, one can understand $\mathcal{W}(V,\hbar\sigma)$ through the standard Schr\"{o}dinger representation $\pi_S$ on $\mathcal{H}_S = L^2(\R^n)$.  In this case, if we let $Q^\hbar_j$ and $P^\hbar_j$ denote the position and momentum operators
\begin{align}
\label{eq:posmom}
    (Q^\hbar_j\psi)(x) &= x_j\cdot\psi(x)\\
    (P^\hbar_j\psi)(x) &= i\hbar\frac{\partial}{\partial x_j} \psi(x)\end{align}
for all $\psi\in L^2(\R)$, then $\pi_S$ is the continuous linear extension of the representation
\begin{align}
\label{eq:srep}
\pi_S(W_\hbar(a,b)) = e^{i\sum_{j=1}^n a_j\cdot P^\hbar_j+ib_j\cdot Q^\hbar_j}
\end{align}
so that $\mathcal{W}(V,\hbar\sigma)$ can be understood as the C*-algebra generated by exponentials of configuration and momentum quantities.

The quantization maps $\mathcal{Q}_\hbar: \Delta(V,0)\to\mathcal{W}(V,\hbar\sigma)$ are given for $\hbar \in [0,1]$ by the linear extension of 
\begin{align}
\label{eq:weylquant}
    \mathcal{Q}_\hbar(W_0(f)) = W_\hbar(f)
\end{align}
for all $f\in V$.  These quantization maps define a strict deformation quantization \citep{BiHoRi04b} on $M = V'$ for the Poisson *-algebra $\mathcal{P} = \Delta(V,0)\subseteq \mathfrak{A}_0 = AP(V')$ and fiber C*-algebras $\mathfrak{A}_\hbar = \mathcal{W}(V,\hbar\sigma)$ for $\hbar>0$.  

One can define a category of classical models with linear phase spaces, as follows.  This category will form the domain of a quantization functor.

\begin{definition}
We denote the following category by $\mathbf{LinClass}$:
\begin{itemize}
    \item \emph{Objects} are pairs $\big(AP(V'), \Delta(V,0)\big)$, where $AP(V')$ is the C*-algebra of almost periodic functions on the dual to a topological vector space $V$, and $\Delta(V,0)$ is the dense Poisson *-subalgebra with Poisson bracket defined by a symplectic form $\sigma$.
    \item \emph{Arrows} are *-homomorphisms $\alpha_0: AP(V')\to AP(U')$ for symplectic topological vector spaces $(V,\sigma)$ and $(U,\sigma')$ that are \emph{smooth} in the sense that 
\begin{align}
\alpha_0\big[\Delta(V,0)\big]\subseteq \Delta(U,0)
\end{align}
and \emph{Poisson} in the sense that \begin{align}
    \alpha_0\big(\{A,B\}_V\big) = \big\{\alpha_0(A),\alpha_0(B)\big\}_U
\end{align}for all $A,B\in \Delta(V,0)$.
\end{itemize}
\end{definition}

\noindent Note that this category is general enough to include infinite-dimensional phase spaces representing linear classical field theories.  The morphisms in this category preserve the structure of classical models at $\hbar = 0$ as symplectic phase spaces.

Similarly, one can define a category of quantum models corresponding to these linear phase spaces.

\begin{definition}
We denote the following category by $\mathbf{LinQuant}$: \begin{itemize}
    \item \emph{Objects} are strict quantizations $\big(\mathcal{W}(V,\hbar\sigma),\Delta(V,\hbar\sigma),\mathcal{Q}_\hbar\big)_{\hbar\in (0,1]}$ of $\mathcal{P} = \Delta(V,0)$, as defined in Eq. (\ref{eq:weylquant}).
    \item \emph{Arrows} are smooth, scaling *-homomorphisms $\alpha_1: \mathfrak{A}_1\to\mathfrak{B}_1$, where $\mathfrak{A}_1 = \mathcal{W}(V,\sigma)$ and $\mathfrak{B}_1 = \mathcal{W}(U,\sigma')$ are the C*-Weyl algebras at $\hbar = 1$ for symplectic topological vector spaces $(V,\sigma)$ and $(U,\sigma')$, respectively.
\end{itemize} 
\end{definition}

\noindent The morphisms in this category thus preserve the structure of the fully quantized models as non-commutative C*-algebras of operators at $\hbar = 1$.

The following proposition characterizing classical morphisms in $\mathbf{LinClass}$ is essential for the definition of the quantization functor.

\begin{prop}
\label{prop:linmorph}
If $\alpha_0: AP(V')\to AP(U')$ is a morphism in $\mathbf{LinClass}$ for symplectic topological vector spaces $V$ and $U$, then there is a character $\chi: V\to \mathbb{C}$ on the additive group $V$ and an additive, symplectic, origin-preserving transformation $T: V\to U$ such that
\begin{align}
\label{eq:classWmorph}
    \alpha_0(W_0(f)) = \chi(f)W_0(Tf)
\end{align}
for each $f\in V$.
\end{prop}

\begin{proof}
First, note that $AP(V')\cong C(\hat{V})$, where $\hat{V}$ is the compact space of all characters on $V$ with the topology of pointwise convergence, and similarly for $AP(U')\cong C(\hat{U})$.  The isomorphism here is defined as follows.  For each $f\in V$, consider the function $\eta(f)\in C(\hat{V})$ defined by
\begin{align}
    \eta(f)(\chi) = \chi(f)
\end{align}
for all $\chi\in\hat{V}$.  Then the isomorphism $AP(V')\cong C(\hat{V})$ is given by the continuous linear extension of
\begin{align}
    W_0(f)\in AP(V')\mapsto \eta(f)\in C(\hat{V})
\end{align}
for all $f\in V$ \citep[Thm. 4-3, p. 2903]{BiHoRi04a}.  In what follows, we will use the symbol $W_0(f)$ for both the element of $AP(V')$ and for the element $\eta(f)\in C(\hat{V})$.

By Gelfand duality \citep[][\S C.2-3]{La17}, there is a unique map $\hat{T}: \hat{U}\to\hat{V}$, continuous in the topologies of pointwise convergence on $\hat{U}$ and $\hat{V}$, such that
\begin{align}
   (\hat{T}\chi)(f) = W_0(f)(\hat{T}\chi) = \alpha_0(W_0(f))(\chi)
\end{align}
for every $f\in V$ and $\chi\in \hat{U}$.  From this, define $\hat{T}': \hat{U}\to\hat{V}$ by
\begin{align}
    \hat{T}'\chi = (\hat{T}\chi)\cdot (\hat{T} e)^{-1}
\end{align}
for all $\chi\in\hat{U}$, where $e\in \hat{U}$ is the identity element of the character group, i.e., the function $e(f) = 1$ for all $f\in U$.  

We further know that $V\cong \hat{\hat{V}}$ (and similarly for $U$), where $\hat{\hat{V}}$ is the group of continuous characters on $\hat{V}$, when the latter is given the topology of pointwise convergence.  The existence of an isomorphism here follows from Pontryagin duality \citep[p. 27]{Ru62} because $V$ is locally compact when considered with the discrete topology (in which $\hat{V}$ is the collection of all continuous characters on $V$), even though $V$ may \emph{not} be locally compact in its original vector space topology.

The identification $V\cong \hat{\hat{V}}$ allows us to define $T: V\to U$ by
\begin{align}
    (Tf)(\chi) = (\hat{T}'\chi)(f) = \frac{(\hat{T}\chi)(f)}{(\hat{T}e)(f)}
\end{align}
for all $f\in V$ and $\chi\in\hat{U}$.

It follows by direct calculation that $T$ is additive, symplectic and origin-preserving, i.e., 
\begin{align}
    T(f+g) &= Tf + Tg\\
    \sigma(f,g) &= \sigma'(Tf,Tg)\\
    T(0) &= 0
\end{align}
for all $f,g\in V$.  Moreover, it follows that for $f\in V$,
\begin{align}
    (\hat{T}e)(f)\cdot W_0(Tf)(\chi) = (\hat{T}\chi)(f) = \alpha_0(W_0(f))(\chi)
\end{align}
for every $\chi\in \hat{U}$.  Hence, we have shown
\begin{align}
    \alpha_0(W_0(f)) = (\hat{T}e)(f) W_0(Tf)
\end{align}
for all $f\in V$.
\end{proof}

\noindent This proposition vindicates the earlier remark that morphisms in $\mathbf{LinClass}$ preserve the structure of classical symplectic phase spaces.  It also provides a way to lift a morphism to any value $\hbar>0$, as follows.

\begin{cor}
\label{cor:weylmorph}
Suppose $\alpha_0: AP(V')\to AP(U')$ is a morphism in $\mathbf{LinClass}$ for symplectic topological vector spaces $(V,\sigma)$ and $(U,\sigma')$.  Suppose $\alpha_0$ is associated with the continuous character $\chi:V\to\mathbb{C}$ and additive, symplectic, origin-preserving transformation $T:V\to U$ as in Eq. (\ref{eq:classWmorph}).  For any $\hbar>0$, define
the map $\alpha_\hbar: \mathcal{W}(V,\hbar\sigma)\to\mathcal{W}(U,\hbar\sigma')$ as the continuous linear extension of
\begin{align}
\label{eq:weylmorph}
    \alpha_\hbar(W_\hbar(f)) = \chi(f)W_\hbar(Tf)
\end{align}
for all $f\in V$.  Then $\alpha_\hbar$ is a *-homomorphism; in particular $\alpha_1$ is smooth and scaling, and thus is a morphism in $\mathbf{LinQuant}$.\footnote{This is a slight generalization of the results in \citep[][\S D]{BiHoRi04b}, which treats the case where $T$ is bijective and linear.  One can easily check that these conditions are not necessary for $\alpha_\hbar$ as given in Eq. (\ref{eq:weylmorph}) to define a *-homomorphism.}
\end{cor}
Hence, quantization defines a functor as follows.

\begin{definition}
The functor $Q_W: \mathbf{LinClass}\to\mathbf{LinQuant}$ is defined by the map on objects
\begin{align}
    \Big(AP(V'),\Delta(V,0)\Big)&\mapsto \Big(\mathcal{W}(V,\hbar\sigma),\Delta(V,\hbar\sigma),\mathcal{Q}_\hbar\Big)_{\hbar\in(0,1]}
\end{align}
where each classical model is associated with its strict deformation quantization via Eq. (\ref{eq:weylquant}), and the map on arrows
\begin{align}
    \Big[\alpha_0: AP(V')\to AP(U')\Big] &\mapsto \Big[\alpha_1: \mathcal{W}(V,\sigma)\to\mathcal{W}(U,\sigma')\Big]
\end{align}
where each morphism is associated with its quantized counterpart via Prop. \ref{prop:linmorph} and Cor. \ref{cor:weylmorph}.
\end{definition}

In the other direction, the classical limit also defines a functor as follows.  Recall that each strict quantization $\big(\mathcal{W}(V,\hbar\sigma),\Delta(V,\hbar\sigma),\mathcal{Q}_\hbar\big)_{\hbar\in (0,1]}$ defines a continuous bundle of C*-algebras with a C*-algebra of continuous sections that we will denote by $\mathcal{W}\big(V,(0,1]\sigma\big)$.  This algebra contains a subalgebra $\Delta\big(V,(0,1]\sigma\big)$ generated by the sections of the form
\begin{align}
    [\hbar\mapsto W_\hbar(f)]
\end{align}
for fixed $f\in V$ and all $\hbar\in (0,1]$.  Moreover, $\mathcal{W}\big(V,(0,1]\sigma\big)$ contains a closed two-sided ideal $K_0(V) = \{a\in \mathcal{W}(V,(0,1]\sigma)\ |\ \lim_{\hbar\to 0}\norm{\phi_\hbar(a)}_\hbar = 0\}$. 

\begin{definition}
The functor $L_W: \mathbf{LinQuant}\to\mathbf{LinClass}$ is defined on objects by
\begin{align}
    \Big(\mathcal{W}(V,\hbar\sigma),\Delta(V,\hbar\sigma),\mathcal{Q}_\hbar\Big)_{\hbar\in (0,1]}\mapsto \Big(\mathcal{W}\big(V,(0,1]\sigma\big)/K_0(V),\  \Delta\big(V,(0,1]\sigma\big)/K_0(V)\Big)
\end{align}
where each quantum model is associated with its classical limit $\mathcal{W}\big(V,(0,1]\sigma\big)/K_0(V) \cong AP(V')$ via the construction surrounding Eq. (\ref{eq:alglim}), and the map on arrows
\begin{align}
    \Big[\alpha_1: \mathcal{W}(V,\sigma)\to \mathcal{W}(U,\sigma')\Big]\mapsto [\alpha_0: AP(V') \to AP(U')\Big]
\end{align}
where each morphism is associated with its classical limit via the construction surrounding Eqs. (\ref{eq:morphlift})-(\ref{eq:morphlim}).
\end{definition}

With these functors now explicitly defined, we have the following result.

\begin{thm}
\label{thm:wequiv}
The functors
\begin{align}
    Q_W: \mathbf{LinClass}\leftrightarrows\mathbf{LinQuant}: L_W
\end{align}
provide a categorical equivalence.
\end{thm}

\begin{proof}
We shall establish the equivalent standard for categorical equivalence provided in \citep[p. 172-3]{Aw10} by defining two natural transformations $\eta: 1_{\mathbf{LinClass}}\to L_W\circ Q_W$ and $\phi: 1_{\mathbf{LinQuant}}\to Q_W\circ L_W$.

To define $\eta$, we recall that it follows from \citep{StFe21a} that for any object $AP(V')$ in $\mathbf{LinClass}$, there is an isomorphism 
\begin{align}
    \eta_V: AP(V')\to \mathcal{W}(V,(0,1]\sigma)/K_0(V) = L_W\circ Q_W(AP(V'))
\end{align}
generated by the linear, continuous extension of
\begin{align}
    \eta_V(W_0(f)) = [\hbar\mapsto \mathcal{Q}_\hbar(W_0(f))] + K_0(V)
\end{align}
for any $f\in V$.  One can easily check that for any arrow $\alpha_0: AP(V')\to AP(U')$ in $\mathbf{LinClass}$, the following diagram commutes:
\begin{equation*}
\begin{tikzcd}
AP(V') \arrow{rr}{\eta_V} \arrow[dd, "\alpha_0"']
& & L_W\circ Q_W(AP(V'))\arrow{dd}{L_W\circ Q_W(\alpha_0)} \\
& & \\
AP(U') \arrow{rr}{\eta_U}
& & L_W\circ Q_W(AP(U'))
\end{tikzcd}
\end{equation*}
\noindent This establishes that $\eta_V$ is a natural isomorphism.

To define $\phi$, we recall that $Q_W(AP(V')) = \mathcal{W}(V,\sigma)$, so we can use the isomorphism
\begin{align}
    \phi_V = Q_W(\eta_V): \mathcal{W}(V,\sigma)\to Q_W\circ L_W(\mathcal{W}(V,\sigma)).
\end{align}
One can easily check that for any arrow $\alpha_1: \mathcal{W}(V,\sigma)\to\mathcal{W}(U,\sigma')$ in $\mathbf{LinQuant}$, the following diagram commutes:
\begin{equation*}
\begin{tikzcd}
\mathcal{W}(V,\sigma)) \arrow{rr}{\phi_V} \arrow[dd, "\alpha_1"']
& & Q_W\circ L_W(\mathcal{W}(V,\sigma))\arrow{dd}{Q_W\circ L_W(\alpha_1)} \\
& & \\
\mathcal{W}(U,\sigma') \arrow{rr}{\phi_U}
& & Q_W\circ L_W(\mathcal{W}(U,\sigma'))
\end{tikzcd}
\end{equation*}
This establishes that $\phi_V$ is a natural isomorphism.
\end{proof}

\noindent This shows that the functors $Q_W$ and $L_W$ provide a one-to-one correspondence between the structure-preserving maps of each model in $\mathbf{LinClass}$ and $\mathbf{LinQuant}$.\footnote{See \citep{BeElYu21} for closely related results on automorphisms of the polynomial Weyl algebras.} Hence, this shows a sense in which, relative to the structure encoded in these choices of categories, classical and quantum models have shared structure, when compared with these choices of functors.

We close this section by emphasizing that the morphisms considered in the categories $\mathbf{LinClass}$ and $\mathbf{LinQuant}$ are significantly constrained by the restriction that they be additive maps between vector spaces, as established in Prop. \ref{prop:linmorph}.  It would be of great interest to extend the categorical equivalence result to a wider class of morphisms.

\section{Rieffel's Quantization for Actions of $\mathbb{R}^d$}
\label{sec:rquant}

While quantization via the Weyl algebra is a prominent example, it has limited applications as well as technical issues.\footnote{For more discussion of issues with the Weyl algebra, see \citep{Gr97,GrNe09,Fe18a,FeMaRoWe19,FeWe19}.} In this section we will consider a quantization prescription for a different algebra.  To do so, we will restrict attention to finite-dimensional phase spaces, so we will lose the generality of the Weyl algebra for representing field theories.  But we will allow ourselves to consider phase spaces that are more generally manifolds and not necessarily linear spaces.

The method of quantization due to Rieffel \citep{Ri89,Ri93} applies to classical systems whose phase space is a manifold $M$ with a diffeomorphic action $\beta$ of the Lie group $\mathbb{R}^d$.  In what follows, we will assume the group $\mathbb{R}^d$ acts freely on $M$.  Furthermore, we assume the Lie group carries a symplectic form $\sigma$ on $\mathbb{R}^d$, which corresponds to an antisymmetric matrix $J_{jk}$ on the vector space $\mathbb{R}^d$, understood as the Lie algebra of the Lie group $\mathbb{R}^d$.  In this case the Poisson *-algebra $C_c^\infty(M)\subseteq C_b(M)$ is the collection of smooth, compactly supported functions on the phase space.  This algebra $C_c^\infty(M)$ is norm dense in the C*-algebra $C_0(M)$ of continuous functions vanishing at infinity on the phase space.  The action of $\mathbb{R}^d$ on $M$ defines an automorphic action $\tau$ of $\mathbb{R}^d$ on $C_0(M)$ by
\begin{align}
\label{eq:action}
    \tau_x(f) = f\circ \beta_x.
\end{align}
The subalgebra $C_c^\infty(M)$ carries a corresponding infinitesimal action of the Lie algebra $\mathbb{R}^d$ by smooth vector fields $\xi_X$ for $X\in \mathbb{R}^d$ given by
\begin{align}
    \xi_X(f) = \frac{\partial}{\partial t}_{|t=0} \tau_{tX}(f)
\end{align}
for all $f\in C_c^\infty(M)$.
The Poisson bracket on $C_c^\infty(M)$ is then defined from the infinitesimal action of the Lie algebra and the symplectic form $\sigma$ for all $f,g\in C_c^\infty(M)$ by
\begin{align}
\label{eq:poissbrack}
    \{f,g\} = \sum_{j,k}J_{jk} \xi_{X_j}(f)\xi_{X_k}(g),
\end{align}
where the vectors $\{X_k\}_{k=1}^d$ form a basis for the Lie algebra $\mathbb{R}^d$.  This definition does not depend on the chosen basis, and one can check that it satisfies the conditions of a Poisson bracket on $M$.  This structure specifies the classical model.\footnote{The methods developed by Rieffel \citep{Ri89} for quantization apply much more broadly, even to deforming products on non-commutative C*-algebras carrying actions of $\mathbb{R}^d$.  The methods have been further generalized by Landsman \citep{La93b,La93a,La98b,La99} to cases where the construction is employed locally, including Riemannian manifolds, principal bundles, and Lie groupoids.  Also, Bieliavsky and Gayral \citep{BiGa15} have provided a generalization of the quantization prescription for a much wider class of group actions.}

The corresponding quantum model is obtained by deforming the product on $C_c^\infty(M)$.  Define $\mathcal{P}_\hbar(M)$ to be the vector space with involution $C_c^\infty(M)$ with the new multiplication operation $\star_\hbar$, sometimes called the Moyal product, defined by
\begin{align}
\label{eq:defprod}
    f_\hbar \star_\hbar g_\hbar = \int_{\mathbb{R}^d}\int_{\mathbb{R}^d} \tau_x(f)\tau_y(g)e^{i\hbar\sigma(x,y)} 
\end{align}
where we use the notation $f_\hbar,g_\hbar\in\mathcal{P}_\hbar(M)$ only to distinguish these from the identical elements $f,g\in C_c^\infty(M)$.  Rieffel \citep{Ri93} shows that this expression can be made well-defined in terms of oscillatory integrals, and that one can define a C*-norm on $\mathcal{P}_\hbar(M)$ so that the completion 
\begin{align*}
    \mathfrak{A}_\hbar(M) = \overline{\mathcal{P}_\hbar(M)}
\end{align*}
with respect to this norm is a C*-algebra.  We note that each C*-algebra $\mathfrak{A}_\hbar(M)$ also carries a strongly continuous group action of $\mathbb{R}^d$, which we denote by $\tau^\hbar$, picked out as the unique continuous extension of the group action $\tau$ on $\mathcal{P}_\hbar$ \citep[See][Thm. 5.11, p. 44]{Ri93}.  Likewise, it follows from \citep[][Thm. 7.1]{Ri93} that there is an infinitesimal action of the Lie algebra, which we denote by $\xi^\hbar$ on the subalgebra $\mathcal{P}_\hbar(M)$.

For example, in the special case where $M = \mathbb{R}^{2n}$ with the group action $\tau$ for $d = 2n$ by translations, we have $\mathfrak{A}_\hbar(\mathbb{R}^{2n}) \cong \mathcal{K}(L^2(\mathbb{R}^n))$.  We can also understand this algebra through the standard Schr\"{o}dinger representation of $\mathfrak{A}_\hbar(\mathbb{R}^{2n})$ on $L^2(\mathbb{R}^n)$, which we now denote $\tilde{\pi}_S$, given by the continuous extension of
\begin{align}
    \tilde{\pi}_S(f_\hbar) = \int_{\mathbb{R}^{2n}} \frac{d^nad^nb}{(2\pi)^n}(\mathcal{F}f)(a,b)\pi_S(W_\hbar(a,b))
\end{align}
for $f_\hbar\in \mathcal{P}_\hbar(\mathbb{R}^{2n})$.  Here, $\mathcal{F}f$ denotes the Fourier transform of the function $f\in C_c^\infty(\mathbb{R}^{2n})$ and $\pi_S(W_\hbar(a,b))$ is the Schr\"{o}dinger representation of the element $W_\hbar(a,b)$ in the Weyl algebra $\mathcal{W}(\mathbb{R}^{2n},\sigma)$ as given by Eq. (\ref{eq:srep}).  

The quantization maps $\mathcal{Q}_\hbar: C_c^\infty(M)\to \mathfrak{A}_\hbar(M)$ are given for $\hbar\in [0,1]$ by
\begin{align}
\label{eq:rquantmap}
    \mathcal{Q}_\hbar(f) = f_\hbar
\end{align}
for all $f\in C_c^\infty(M)$.  These quantization maps define a strict deformation quantization on $M$ with $\mathcal{P} = C_c^\infty(M)$ and fiber algebras $\mathfrak{A}_\hbar(M)$ for $\hbar>0$.

To define our categories of classical and quantum models suitable for Rieffel quantization, we will need to specify when a morphism of a C*-algebra (either $C_0(M)$ or $\mathfrak{A}_\hbar(M)$) is compatible with a group action.  Suppose we have a *-homomorphism $\alpha: \mathfrak{A}\to\mathfrak{B}$ between two C*-algebras $\mathfrak{A}$ and $\mathfrak{B}$ carrying group actions by $\mathbb{R}^d$ and $\mathbb{R}^{d'}$, respectively.  We now denote the infinitesimal action of the Lie algebra by $\xi$ (corresponding to the action $\xi$ or $\xi^\hbar$, as above.)  We will call $\alpha$ \emph{compatible with the group actions} if for each $X\in\mathbb{R}^{d'}$, there is a $Y\in \mathbb{R}^d$ such that $\xi_X\circ\alpha = \alpha\circ \xi_Y$ on the domain of $\xi_X$ and $\xi_Y$ (i.e., on $C_c^\infty(M)$ or $\mathcal{P}_\hbar$).

Now we define a category of classical models suitable for Rieffel quantization.

\begin{definition}
We denote the following category by $\mathbf{RClass}$: \begin{itemize}
\item \emph{Objects} are triples $(C_0(M), C^\infty_c(M),\tau)$, where $C_0(M)$ is the C*-algebra of continuous functions vanishing at infinity on a manifold $M$, and $C_c^\infty(M)$ is a dense *-subalgebra.  Further $\tau$ is a strongly continuous, free action of $\mathbb{R}^d$ on $C_0(M)$ arising from a family of diffeomorphisms by Eq. (\ref{eq:action}), which defines the Poisson bracket on $C_c^\infty(M)$ by Eq. (\ref{eq:poissbrack}).
\item \emph{Arrows} are *-homomorphisms $\alpha_0: C_0(M)\to C_0(N)$ for manifolds $M$ and $N$ that are compatible with the group actions, that are \emph{smooth} in the sense that 
\begin{align}
    \alpha_0\big[C_c^\infty(M)\big]\subseteq C_c^\infty(N)
\end{align}
and that are \emph{Poisson} in the sense that
\begin{align}
    \alpha_0\big(\{A,B\}_M\big) = \big\{\alpha_0(A),\alpha_0(B)\big\}_N
\end{align}
for all $A,B\in C_c^\infty(M)$.
\end{itemize}
\end{definition}
\noindent Note that this category is general enough to include non-linear phase spaces.  The morphisms in this category preserve the Poisson structure of classical models at $\hbar = 0$ as phase spaces.

Similarly, we define a category of quantum models corresponding to Rieffel's quantization prescription

\begin{definition}
We denote the following category by $\mathbf{RQuant}$: \begin{itemize}
\item \emph{Objects} are strict quantizations $\big(\mathfrak{A}_\hbar(M), \mathcal{P}_\hbar(M),\mathcal{Q}_\hbar, \tau^\hbar\big)_{\hbar\in (0,1]}$ of $\mathcal{P} = C_c^\infty(M)$, as given by the discussion around Eq. (\ref{eq:defprod}) and (\ref{eq:rquantmap}).
\item \emph{Arrows} are smooth, scaling *-homomorphisms $\alpha_1: \mathfrak{A}_1(M)\to\mathfrak{A}_1(N)$ that are compatible with the group actions, where $\mathfrak{A}_1(M)$ and $\mathfrak{A}_1(N)$ are the C*-algebras at $\hbar = 1$ obtained as the Rieffel quantizations of $C_c^\infty(M)$ and $C_c^\infty(N)$, respectively.
\end{itemize}
\end{definition}
\noindent The morphisms in this category preserve the structure of the fully quantized models as non-commutative C*-algebras of operators at $\hbar = 1$.

The following proposition characterizing classical morphisms in $\mathbf{RClass}$ is essential for the definition of the quantization functor.

\begin{prop}
\label{prop:smoothmorph}
If $\alpha_0: C_0(M)\to C_0(N)$ is a morphism in $\mathbf{RClass}$ for manifolds $M$ and $N$ with actions of $\mathbb{R}^d$ and $\mathbb{R}^{d'}$, respectively, then there is a smooth map $\varphi: N\to M$ such that
\begin{align}
    \alpha_0(f) = f\circ\varphi
\end{align}
for each $f\in C_0(M)$, and a symplectic map $S: \mathbb{R}^{d'}\to\mathbb{R}^d$ such that
\begin{align}
\label{eq:equivariance}
    \alpha_0\circ\tau_{SX} = \tau_X\circ\alpha_0
\end{align}
for all $X\in \mathbb{R}^{d'}$.
\end{prop}

\begin{proof}
It follows from Gelfand duality \citep[][\S C.2-3]{La17} that there is a continuous map $\varphi: N\to M$ such that
\begin{align}
\alpha_0(f) = f\circ \varphi
\end{align}
for all $f\in C_0(M)$.

To show that $\varphi$ is smooth, we claim that whenever $f\in C^\infty(M)$, it follows that $f\circ\varphi\in C^\infty(N)$.  To show this, we suppose $f\in C^\infty(M)$ and aim to show that $f\circ \varphi$ is smooth in an open neighborhood of each $p\in N$.  For any such $p\in N$, we know there is some open neighborhood $O$ of $\varphi(p)\in M$ and some function $\tilde{f}\in C_c^\infty(M)$ such that $\tilde{f}_{|O} = f_{|O}$.  Therefore, since $\alpha_0$ is smooth, we have that $\tilde{f}\circ\varphi = \alpha_0(\tilde{f})\in C_c^\infty(N)$.  Moreover, it follows that $\varphi^{-1}[O]$ is an open neighborhood of $p\in N$ and $(\tilde{f}\circ\varphi)_{|\varphi^{-1}[O]} = (f\circ\varphi)_{|\varphi^{-1}[O]}$.  Hence, $f\circ\varphi$ is smooth in an open neighborhood of $p$, and since $p\in N$ was arbitrary, it follows that $f\circ\varphi$ is smooth on $N$.  Finally, it follows from the fact that $f\in C^\infty(M)$ was arbitrary that $\varphi$ is smooth.\footnote{This is a small extension of the well-known fact known as ``Milnor's exercise" \citep[Cor. 35.10]{KoMiSl93}.  We have simply extended the correspondence between algebra homomorphisms and smooth maps from algebras of the type $C^\infty(M)$ to algebras of the type $C_c^\infty(M)$ for a manifold $M$.}

To show that $\varphi$ corresponds to a symplectic map $S:\mathbb{R}^{d'}\to\mathbb{R}^d$, we draw on the fact that $\alpha_0$ is Poisson, meaning that
\begin{align}
    \alpha_0\big(\{f,g\}_M\big) = \big\{\alpha_0(f),\alpha_0(g)\big\}_N
\end{align}
for all $f,g\in C_c^\infty(M)$, which implies
\begin{align}
    \{f,g\}_M\circ\varphi = \{f\circ\varphi,g\circ\varphi\}_N.
\end{align}
Given $X\in \mathbb{R}^{d'}$, since $\alpha_0$ is compatible with the group actions, there is an element $SX\in \mathbb{R}^{d}$ satisfying
\begin{align}
\label{eq:comp}
\xi_X\circ\alpha_0 = \alpha_0\circ\xi_{SX}.
\end{align}
Moreover the value $SX\in\mathbb{R}^d$ in Eq. (\ref{eq:comp}) is unique since the group actions are free.  So this can be rewritten as a local condition that at each $p\in N$, where we denote the differential of $\varphi$ at $p\in N$ by $d\varphi_p: T_pN\to T_{\varphi(p)}M$.  We have that for each $f\in C_c^\infty(M)$,
\begin{align}
    d\varphi_p\big((\xi_X)_{|p}\big)(f) &= (\xi_X)_{|p}(f\circ\varphi)= (\xi_X)_{|p}(\alpha_0(f)) = \big((\xi_X\circ\alpha_0)(f)\big)(p)\\
    &= \big((\alpha_0\circ \xi_{SX})(f)\big)(p) = \big((\xi_{SX})(f)\big)\big(\varphi(p)\big) = (\xi_{SX})_{|\varphi(p)}(f).
\end{align}
In other words, we have that $\xi_X$ and $\xi_{SX}$ are $\varphi$-related.  This implies that we have defined a linear map $S: \mathbb{R}^{d'}\to \mathbb{R}^d$ for all $X\in\mathbb{R}^{d'}$.  The fact that $S$ is linear follows from the linearity of the differential $d\varphi_p$ at each point $p\in N$ and the linearity of the map $X\mapsto \xi_X$.  It now follows from the definition of the Poisson bracket that $S$ is symplectic.  Moreover, it follows from the compatibility with the group actions that
\begin{align}
    \alpha_0\circ\tau_{SX} = \tau_X \circ \alpha_0,
\end{align}
as desired.
\end{proof}

This proposition vindicates the earlier remark that morphisms in $\mathbf{RClass}$ preserve the structure of classical phase spaces.  It also provides a way to lift a morphism to any value $\hbar>0$, as follows.  Since the map $\varphi$ corresponding to $\alpha_0$ is a smooth map, it follows from the compatibility of $\alpha_0$ with the group action that the pushforward $\varphi_*$ lifts to a symplectic transformation $S: \mathbb{R}^{d'}\to\mathbb{R}^{d}$.  It then follows from Prop. 10.4 of \citep[][p. 70]{Ri93} that each morphism $\alpha_0$ can be quantized.

\begin{cor}
\label{cor:quantmorph}
If $\alpha_0: C_0(M)\to C_0(N)$ is a morphism in $\mathbf{RClass}$ for manifolds $M$ and $N$, then the restriction of $\alpha_0$ to $\mathcal{P}_\hbar(M) = C_c^\infty(M)$ continuously extends to a *-homomorphism $\alpha_\hbar: \mathfrak{A}_\hbar(M)\to\mathfrak{A}_\hbar(N)$ for each $\hbar>0$.  Moreover, the map $\alpha_1$ so defined is smooth and scaling, and thus is a morphism in $\mathbf{RQuant}$.\footnote{The fact that $\alpha_1$, so defined, is compatible with the group action follows from \citep[][Thm. 7.1]{Ri93}.}
\end{cor}

Hence, quantization defines a functor as follows.

\begin{definition}
The functor $Q_R: \mathbf{RClass}\to\mathbf{RQuant}$ is defined on objects by
\begin{align}
    \Big(C_0(M), C_c^\infty(M),\tau\Big) \mapsto \Big(\mathfrak{A}_\hbar(M),\mathcal{P}_\hbar(M),\mathcal{Q}_\hbar,\tau^\hbar\Big)_{\hbar\in (0,1]}
\end{align}
where each classical model is associated with its strict deformation quantization via Eq. (\ref{eq:rquantmap}), and the map on arrows
\begin{align}
    \Big[\alpha_0: C_0(M)\to C_0(N)\Big]\mapsto \Big[\alpha_1: \mathfrak{A}_1(M)\to\mathfrak{A}_1(N)\Big]
\end{align}
where each morphism is associated with its quantized counterpart via Prop. \ref{prop:smoothmorph} and Cor. \ref{cor:quantmorph}.
\end{definition}

In the other direction, the classical limit also defines a functor as follows.  Recall that each strict quantization $\big(\mathfrak{A}_\hbar(M),\mathcal{P}_\hbar(M),\mathcal{Q}_\hbar\big)_{\hbar\in (0,1]}$ defines a continuous bundle of C*-algebras with a C*-algebra of continuous sections that we will denote by $\mathfrak{A}_{(0,1]}(M)$.  The algebra contains a subalgebra $\mathcal{P}_{(0,1]}(M)$ generated by the sections of the form
\begin{align}
    [\hbar\mapsto f_\hbar]
\end{align}
for fixed $f\in C_c^\infty(M)$ and all $\hbar\in (0,1]$.  Moreover, $\mathfrak{A}_{(0,1]}(M)$ contains a closed two-sided ideal $K_0(M) = \{a\in\mathfrak{A}_{(0,1]}(M)\ | \lim_{\hbar\to 0}\norm{\phi_\hbar(a)}_\hbar = 0\}$.

\begin{definition}
The functor $L_R:\mathbf{RQuant}\to\mathbf{RClass}$ is defined on objects by
\begin{align}
    \Big(\mathfrak{A}_\hbar(M),\mathcal{P}_\hbar(M),\mathcal{Q}_\hbar,\tau^\hbar\Big)_{\hbar\in (0,1]}\mapsto \Big(\mathfrak{A}_{(0,1]}(M)/K_0(M),\mathcal{P}_{(0,1]}(M)/K_0(M),\tau^0\Big)
\end{align}
where each quantum model is associated with its classical limit $\mathfrak{A}_{(0,1]}(M)/K_0(M)\cong C_0(M)$ via the construction surrounding Eq. (\ref{eq:alglim}) and where $\tau^0$ is the classical limit of the scaling morphism $\tau^1$, which is specified in the same way the functor maps all arrows.  The functor acts on arrows by
\begin{align}
    \Big[\alpha_1: \mathfrak{A}_1(M)\to\mathfrak{A}_1(N)\Big]\mapsto\Big[\alpha_0: C_0(M)\to C_0(N)\Big]
\end{align}
where each morphism is associated with its classical limit via the construction surrounding Eqs. (\ref{eq:morphlift})-(\ref{eq:morphlim}). 
 Steeger and Feintzeig \citep[][Prop. 5.5]{StFe21a} establish that $\alpha_0$ so defined is Poisson, and hence is a morphism in $\mathbf{RClass}$.\footnote{The fact that $\alpha_0$, so defined, is compatible with the group action follows from \citep[][Thm. 7.1]{Ri93}.}
\end{definition}

With these functors now explicitly defined, we have the following result.

\begin{thm}
\label{thm:requiv}
The functors
\begin{align}
    Q_R:\mathbf{RClass}\leftrightarrows\mathbf{RQuant}: L_R
\end{align}
provide a categorical equivalence.
\end{thm}

\begin{proof}
The proof proceeds exactly as for Thm. \ref{thm:wequiv}.  We shall establish the equivalent standard for categorical equivalence provided in \citep[p. 172-3]{Aw10} by defining two natural transformations $\eta: 1_{\mathbf{RClass}}\to L_R\circ Q_R$ and $\phi: 1_{\mathbf{RQuant}}\to Q_R\circ L_R$.

To define $\eta$, we recall that it follows from \citep{StFe21a} that for any object $C_0(M)$ in $\mathbf{RQuant}$, there is an isomorphism 
\begin{align}
    \eta_M: C_0(M)\to \mathfrak{A}_{(0,1]}(M)/K_0(M) = L_R\circ Q_R(C_0(M))
\end{align}
generated by the linear, continuous extension of
\begin{align}
    \eta_V(f) = [\hbar\mapsto \mathcal{Q}_\hbar(f)] + K_0(M)
\end{align}
for any $f\in C_c^\infty(M)$.  One can easily check that for any arrow $\alpha_0: C_0(M)\to C_0(N)$ in $\mathbf{RClass}$, the following diagram commutes:
\begin{equation*}
\begin{tikzcd}
C_0(M) \arrow{rr}{\eta_M} \arrow[dd, "\alpha_0"']
& & L_R\circ Q_R(C_0(N))\arrow{dd}{L_W\circ Q_W(\alpha_0)} \\
& & \\
C_0(N) \arrow{rr}{\eta_N}
& & L_R\circ Q_R(C_0(N))
\end{tikzcd}
\end{equation*}
\noindent This establishes that $\eta_V$ is a natural isomorphism.

To define $\phi$, we recall that $Q_R(C_0(M)) = \mathfrak{A}_1(M)$, so we can use the isomorphism
\begin{align}
    \phi_M = Q_R(\eta_M): \mathfrak{A}_1(M)\to Q_R\circ L_R(\mathfrak{A}_1(M)).
\end{align}
One can easily check that for any arrow $\alpha_1: \mathfrak{A}_1(M)\to\mathfrak{A}_1(N)$ in $\mathbf{RQuant}$, the following diagram commutes:
\begin{equation*}
\begin{tikzcd}
\mathfrak{A}_1(M) \arrow{rr}{\phi_M} \arrow[dd, "\alpha_1"']
& & Q_R\circ L_R(\mathfrak{A}_1(M))\arrow{dd}{Q_R\circ L_R(\alpha_1)} \\
& & \\
\mathfrak{A}_1(N) \arrow{rr}{\phi_N}
& & Q_R\circ L_R(\mathfrak{A}_1(N))
\end{tikzcd}
\end{equation*}
This establishes that $\phi_V$ is a natural isomorphism.
\end{proof}

\noindent Just as in the previous section, this shows that the functors $Q_R$ and $L_R$ provide a one-to-one correpsondence between the structure-preserving maps of each model in $\mathbf{RClass}$ and $\mathbf{RQuant}$.  Hence, this shows a sense in which, relative to the structure encoded in these choices of categories, classical and quantum models share structure, when compared with these choices of functors.

We close this section by emphasizing that the morphisms considered in the categories $\mathbf{RClass}$ and $\mathbf{RQuant}$ are significantly constrained by the restriction that they satisfy the condition of Eq. (\ref{eq:equivariance}), which was established in Prop. \ref{prop:smoothmorph}.  We can think of this condition as the requirement that morphisms are smooth maps that are ``almost equivariant" for the group actions.  It would be of great interest to extend the categorical equivalence result to a wider class of morphisms.

\section{Conclusion}
\label{sec:con}

We have established that in two cases---the quantization of linear spaces using the C*-Weyl algebra and the quantization of phase spaces with actions of $\mathbb{R}^d$ through Rieffel's prescription---quantization and the classical limit each define functors that together form a categorical equivalence.  This shows a precise sense in which for each of these cases, the algebraic structure of observables in classical physics corresponds with the algebraic structure of observables in quantum physics.

The models and quantization procedures we have used are among the simplest mathematically, so the present results serve primarily as a proof of concept.  There are open questions concerning whether the results can be extended either to a broader class of objects or to a broader class of arrows.  Can the results concerning Rieffel's quantization be extended to quantization by actions of non-Abelian groups via the prescription of Bieliavsky and Gayral \citep{BiGa15}?  Can the results concerning quantization of *-homomorphisms be extended to the quantization of more general Hilbert bimodules from symplectic dual pairs via the prescription of Landsman \citep{La01a,La02a}?  We leave these questions for future research. We hope that the current paper establishes merely the possibility for positive results and provides reason to be interested in quantization as a categorical equivalence.

\section*{Acknowledgements}
Funding was provided by the National Science Foundation (Grant No.
571 2043089).

\section*{Conflict of Interest Statement}

On behalf of all authors, the corresponding author states that there is no conflict of interest.

\section*{Data Availability Statement}

Data sharing not applicable to this article as no datasets were generated or analysed during the current study.

\newpage
\bibliographystyle{plain}
\bibliography{bibliography.bib}

\end{document}